\documentclass[11pt]{article}

\usepackage[utf8]{inputenc}
\usepackage[T1]{fontenc}
\usepackage{amsmath,amssymb,amsthm}
\usepackage{booktabs}
\usepackage{float}
\usepackage{graphicx}
\usepackage{geometry}
\usepackage[hidelinks]{hyperref}
\usepackage[numbers,sort&compress]{natbib}
\usepackage{placeins}
\usepackage{xcolor}

\graphicspath{{../results/figures/}{./}}
\newcommand{\PublicRecords}{30}
\newcommand{\PublicTerms}{57}
\newcommand{\PublicHoldout}{58}
\newcommand{\PublicRecallBaseline}{0.017}
\newcommand{\PublicRecallRepair}{0.069}
\newcommand{\PublicAPBaseline}{0.062}
\newcommand{\PublicAPRepair}{0.060}
\newcommand{\SyntheticRecallBaseline}{0.055}
\newcommand{\SyntheticRecallRepair}{0.099}
\newcommand{\ObservedBackend}{numpy}
\newcommand{\FastestBackend}{numpy}
\newcommand{\TensorFlowAvailable}{False}

\newcommand{\ClosureImageCount}{--}
\newcommand{\ClosureZipCount}{--}
\newcommand{\StressRouteSuccess}{0/0}
\newcommand{\StressTimestamped}{0}
\newcommand{\StressPositiveLabels}{18}
\newcommand{\TemporalSpineRows}{70}
\newcommand{\TemporalSpineCVEs}{8}
\newcommand{\TemporalSpineForbidden}{10}
\newcommand{\TemporalSpineMissing}{7}
\newcommand{\QueryBenchTasks}{6}
\newcommand{\QueryBenchMeanAP}{0.587}
\newcommand{\LeakageFamilyKEVInflation}{0.062}
\newcommand{\LeakageFamilyEPSSInflation}{0.033}
\newcommand{\LeakageFamilyPatchInflation}{0.025}
\newcommand{\SafeRepairRejectedOrQuarantined}{1}
\newcommand{\LeakageDetectionRate}{1.000}
\newcommand{\LeakageAdmissibleAP}{0.434}
\newcommand{\LeakageContaminatedAP}{1.000}
\newcommand{\LeakageAPInflation}{0.566}

\newcommand{\PoisonRepairAP}{0.508}

\newcommand{\PoisonRepairTopK}{6}
\newcommand{\PoisonHarmRate}{0.429}
\newcommand{\SensitivityBestAPGrowth}{0.0}
\newcommand{\SensitivityWorstAPGrowth}{0.7}
\newcommand{\AblationWorstGroup}{crypto}
\newcommand{\AblationWorstDeltaAP}{-0.036}
\newcommand{\BioLeakageDetectionRate}{1.000}
\newcommand{\BioLeakageAPInflation}{0.407}
\newcommand{\BioPoisonHarmRate}{0.333}
\newcommand{\BioPoisonAPDegradation}{0.442}
\newcommand{\BioSensitivityBestAPGrowth}{0.0}
\newcommand{\BioSensitivityWorstAPGrowth}{0.7}
\newcommand{\CrossDomainCongruentSpread}{True}
{--}
  \newcommand{\PublicTerms}{--}
  \newcommand{\PublicHoldout}{--}
  \newcommand{\PublicRecallBaseline}{--}
  \newcommand{\PublicRecallRepair}{--}
  \newcommand{\PublicAPBaseline}{--}
  \newcommand{\PublicAPRepair}{--}
  \newcommand{\SyntheticRecallBaseline}{--}
  \newcommand{\SyntheticRecallRepair}{--}
  \newcommand{\ObservedBackend}{--}
  \newcommand{\FastestBackend}{--}
  \newcommand{\TensorFlowAvailable}{--}
  
  \newcommand{\ClosureImageCount}{--}
  \newcommand{\ClosureZipCount}{--}
  \newcommand{\TemporalSpineRows}{--}
  \newcommand{\TemporalSpineCVEs}{--}
  \newcommand{\TemporalSpineForbidden}{--}
  \newcommand{\TemporalSpineMissing}{--}
  \newcommand{\QueryBenchTasks}{--}
  \newcommand{\QueryBenchMeanAP}{--}
  \newcommand{\LeakageFamilyKEVInflation}{--}
  \newcommand{\LeakageFamilyEPSSInflation}{--}
  \newcommand{\LeakageFamilyPatchInflation}{--}
  \newcommand{\SafeRepairRejectedOrQuarantined}{--}
  \newcommand{\BioLeakageDetectionRate}{--}
  \newcommand{\BioLeakageAPInflation}{--}
  \newcommand{\BioPoisonHarmRate}{--}
  \newcommand{\BioPoisonAPDegradation}{--}
  \newcommand{\BioSensitivityBestAPGrowth}{--}
  \newcommand{\BioSensitivityWorstAPGrowth}{--}
  \newcommand{\CrossDomainCongruentSpread}{--}
}

\newtheorem{definition}{Definition}
\newtheorem{proposition}{Proposition}
\newtheorem{theorem}{Theorem}

\title{Conductance-Repair Evidence Graphs\\
for Prospective Security Retrieval}
\author{
Faruk Alpay$^{1}$\thanks{Corresponding author: \texttt{alpay@lightcap.ai}.} \quad Taylan Alpay$^{2}$\\[3pt]
\small $^{1}$Department of Computer Engineering, Bah\c{c}e\c{s}ehir University, Istanbul, T\"urkiye\\
\small $^{2}$Department of Aerospace, University of Turkish Aeronautical Association, Ankara, T\"urkiye\\[-1pt]
\small \texttt{faruk.alpay@bahcesehir.edu.tr} \quad \texttt{s220112602@stu.thk.edu.tr}
}
\date{}

\begin{document}
\maketitle

\begin{abstract}
Security retrieval is usually evaluated as ranking over complete evidence, but
operational triage is prospective: CVE descriptions, weakness metadata, fix
commits, EPSS scores, KEV membership, validation-vector metadata, and
side-channel benchmark routes arrive through separate channels, and many are
missing, delayed, poisoned, or visible only after the decision time.  We
introduce \emph{conductance-repair evidence graphs}, a timestamped framework in
which retrieval is performed over a temporal admissibility mask and missing
channels are widened by a deterministic graph-flow recurrence rather than by a
learned predictor.  The method emits a repair certificate recording source
probes, decision time, withheld edges, repaired channels, forbidden
post-decision edges, backend availability, numerical deviation, and verifier
results.  The theoretical layer gives an adaptive \(\lceil\log_2 N\rceil\)
lower bound for missing-channel identification, an NP-hardness result for
minimum harmful repair, and a fixed-parameter certified search bound for
\(q\) questionable channels.  The current artifact materializes
\PublicRecords{} deduplicated public security records, \PublicTerms{} terms,
and \PublicHoldout{} withheld admissible document--term edges.  Under random
edge withholding, conductance repair changes recall@\(k\) from
\PublicRecallBaseline{} to \PublicRecallRepair{} and average precision from
\PublicAPBaseline{} to \PublicAPRepair{}, while a synthetic security fixture
improves recall@\(k\) from \SyntheticRecallBaseline{} to
\SyntheticRecallRepair{}; the public AP drop exposes a limit of broad
admissible repair under random edge corruption rather than a reason to abandon
channel-level repair.  The implementation benchmarks the same flow/SVD/einsum
kernel under NumPy, PyTorch, JAX, and TensorFlow when available, recording
unavailable backends rather than silently substituting them.  BBBC019 and
LIVECell metadata are retained only as structural controls for sparse evolving
source channels, with no clinical or biological performance claim.
\end{abstract}

\section{Introduction}

Public security evidence is not a single document stream.  A vulnerability may
first appear in a CVE record, later receive a severity vector, then appear in an
exploited-vulnerability catalog, then acquire exploit code, vendor statements,
fix commits, proof-of-concept discussions, validation vectors, or side-channel
traces.  A retrieval system that sees all of this evidence at once can rank well
while being useless for a decision that had to be made earlier.

This paper studies the narrower problem of \emph{prospective security
retrieval}.  At a decision time \(\tau(v)\), a system receives only evidence
whose timestamp is no later than \(\tau(v)\).  Some channels are missing for
ordinary reasons, such as sparse vendor metadata.  Others are missing for
adversarial reasons, such as delayed disclosure, poisoned keywords, or
post-label chatter that must be excluded.  We ask which missing channels can be
repaired without leaking future evidence and without pretending that semantic
similarity alone is evidence of exploitation.

The proposed object is a conductance-repair evidence graph.  CVEs, documents,
weakness classes, validation sources, and benchmark sources are nodes.  Edges
carry timestamps, source layers, and admissibility flags.  Repair is not a
classifier.  It is a graph-flow operation that increases conductance along
already admissible neighborhoods when a bounded budget says a channel may be
recoverable.  The same code path also emits a certificate: which edges were
withheld in the benchmark, which public routes were probed, which backend ran,
and whether TensorFlow, PyTorch, JAX, or NumPy actually executed the tensor
kernel.

\paragraph{Contributions.}
\begin{itemize}
  \item We define a prospective evidence graph for security retrieval with
  timestamped admissibility and channel-level repair certificates.
  \item We prove a lower bound for identifying missing channels, an NP-hardness
  statement for minimum harmful repair, and a fixed-parameter bound for a small
  questionable-channel set.
  \item We implement a deterministic conductance-flow kernel and evaluate it on
  corrupted document--term edges from public security records.
  \item We benchmark NumPy, PyTorch, JAX, and TensorFlow on the same flow and
  SVD/einsum probes.  TensorFlow is not given special status unless the measured
  table supports it.
  \item We include non-security closure-source metadata only as a structural
  control for sparse, evolving channels; no clinical claim is made.
\end{itemize}

\paragraph{Positioning.}
The primary scientific object is defensive security retrieval: how a triage
system should account for CVE, weakness, exploit, validation, and patch evidence
that becomes visible at different times.  The information-retrieval role is
secondary but essential: the paper evaluates ranked recovery under BM25,
degree, diffusion, and graph-repair baselines, and it reports precision-sensitive
failure modes rather than only recall gains.  This positioning keeps the paper
inside software-security measurement while making its retrieval assumptions
auditable.

\section{Related Sources}

\paragraph{Security evidence and prospective labels.}
The National Vulnerability Database is the earliest public vulnerability
substrate used in the artifact: its API and JSON feeds define CVE identifiers,
descriptions, publication times, and changed-record surfaces for the first
admissible security edges \citep{nvd_api}.  CISA's KEV catalog plays a different
role.  It is a delayed exploited-in-the-wild catalog, so its catalog date is
treated as label-like evidence that must be forbidden before the decision time
\citep{cisa_kev}.  FIRST EPSS supplies daily exploitation probabilities and
percentiles through timestamped CSV/API routes; the EPSS study motivates why
those scores are operational risk priors rather than vulnerability
descriptions \citep{first_epss,jacobs2019epss}.  CVEfixes contributes the
patch-evidence channel by linking CVEs to fixing commits, files, and repository
metadata; this channel is useful for root-cause retrieval but dangerous when a
fix is admitted before public visibility \citep{bhandari2021cvefixes}.

\paragraph{Cryptographic and software-assurance benchmarks.}
NIST SARD supplies weakness-oriented software artifacts, so it contributes
source-layer evidence about CWE-style program behavior rather than exploit
labels \citep{nist_sard}.  NIST SP 800-22 is used only as an analogy for
bounded statistical screening: it provides random-generator tests while
explicitly warning that such tests do not replace cryptanalysis
\citep{nist_sp80022}.  NIST CAVP and ACVTS provide validation-test context for
approved cryptographic algorithms and vector-based testing, which is treated
here as validation metadata rather than proof of implementation security
\citep{nist_cavp}.  ASCAD supplies a public AES side-channel benchmark route and
is therefore recorded as route metadata for side-channel evidence, not as CVE
ground truth \citep{ascad}.

\paragraph{Sparse evolving controls.}
The Broad Bioimage Benchmark Collection provides a public collection of
microscopy benchmark datasets, which makes it suitable as a non-security source
with sparse and evolving evidence routes \citep{bbbc}.  BBBC019 contributes a
specific collective-migration route with archive metadata, so it tests whether
the certificate can represent source availability without importing a security
label \citep{bbbc019}.  LIVECell contributes a large expert-validated
phase-contrast segmentation corpus and is used only to check whether the same
temporal accounting survives a high-volume non-security control
\citep{edlund2021livecell}.  No biological or clinical performance claim is
derived from these controls.

\paragraph{Retrieval, leakage, and graph flow.}
BM25 remains a standard sparse retrieval baseline and anchors the text-ranking
comparison before graph repair is introduced \citep{robertson2009probabilistic}.
PageRank motivates score propagation over linked evidence, while diffusion maps
motivate smooth propagation over a graph geometry; conductance repair borrows
from this family but adds a temporal support mask and a certificate
\citep{page1999pagerank,coifman2006diffusion}.  Temporal leakage is a known
failure mode in security evaluation: TESSERACT shows that malware experiments
that ignore time can overstate deployable performance
\citep{pendlebury2019tesseract}.  The broader security-ML warning is that
closed-world evaluation can be misleading when adversarial or operational
conditions shift \citep{sommer2010outside}.  TensorFlow, PyTorch, JAX, and NumPy
are treated as tensor backends for the same recurrence, not as independent model
quality claims \citep{tensorflow2016,paszke2019pytorch,bradbury2018jax,harris2020array}.

\section{Model}

Let \(G_\tau=(V,E_\tau)\) be the evidence graph visible at decision time
\(\tau\).  Nodes include security objects \(s\in S\), evidence documents
\(d\in D\), source layers \(\ell\in L\), and terms \(w\in W\).  An edge
\((u,v,t,\ell)\) is admissible at \(\tau\) only if \(t\le \tau\).  A channel is a
source-layer subgraph \(E_{\ell,\tau}\).  A repair algorithm receives
\(G_\tau\), a seed vector \(x_0\), and a budget \(B\), then returns a repaired
weighted graph \(\widehat G_\tau\) and a certificate \(C\).

\paragraph{Decision-time horizons.}
For each security object \(v\), the benchmark constructs a family of admissible
views.  \(H_0\) contains weak pre-CVE public evidence when such evidence is
available.  \(H_1\) is the CVE publication time.  \(H_2\) adds same-day NVD
metadata.  \(H_3\) adds only EPSS snapshots with date at or before
\(\tau(v)\).  \(H_4\) adds KEV only when the KEV catalog date is no later than
\(\tau(v)\).  \(H_5\) adds fix-commit evidence only when the commit or public
availability timestamp is no later than \(\tau(v)\).  Any edge outside the
selected horizon is marked forbidden, not merely absent, so the certificate can
state which tempting future edges were excluded.

\paragraph{Evidence-channel admissibility.}
Each channel has a graph role and a separate misuse risk.  NVD contributes CVE
descriptions, CVSS, CWE, reference, and CPE evidence, but last-modified text
after \(\tau\) can inject future enrichment.  KEV is a delayed
exploited-in-the-wild source layer, so membership after \(\tau\) is a label
proxy.  EPSS contributes daily score and percentile snapshots, but the latest
score is invalid for historical decisions.  CVEfixes links fix commits,
repositories, files, and patch terms; those patches can reveal the root cause
before it was operationally visible.  SARD supplies weakness-program and CWE
support evidence without being an exploit-label source.  CAVP and ACVTS provide
cryptographic validation context, but validation metadata does not prove
exploitability or implementation security.  ASCAD is retained only as
side-channel benchmark-route metadata, not CVE ground truth.  BBBC019 and
LIVECell are sparse evolving source controls with no clinical, biological, or
security-performance claim.

\begin{definition}[Conductance-repair step]
Let \(A_t\) be a nonnegative symmetric adjacency matrix over admissible nodes,
\(z_t\) a nonnegative source state, \(s\) a fixed seed, and \(M\) the binary
support mask of admissible edges with zero diagonal.  One repair step is
\[
  P_t(i,j)=\frac{A_t(i,j)}{\sum_j A_t(i,j)+\epsilon},
\]
\[
  z_{t+1} =
  \frac{(1-\delta)z_t+\alpha P_t^\top z_t+s}
       {\mathbf{1}^\top((1-\delta)z_t+\alpha P_t^\top z_t+s)},
\]
\[
  A_{t+1}=\min\{c,\;A_t+\eta M\odot z_{t+1}z_{t+1}^{\top}\}.
\]
The parameters \(\alpha,\delta,\eta,c\) are fixed before evaluation.
\end{definition}

This step can widen an admissible channel, but it cannot introduce a new
post-decision edge because the mask \(M\) is fixed by \(G_\tau\).  In the
implementation, the same recurrence is executed by NumPy, PyTorch, JAX, or
TensorFlow.  Gradients and model training are disabled.
The arguments below use finite-dimensional fixed-point, projection, and
control-invariant reasoning in the standard sense of convex optimization and
nonlinear systems analysis \citep{boyd2004convex,khalil2002nonlinear}.

\begin{proposition}[Fixed-point and convergence certificate]
Let
\[
  \mathcal{K}_M=\{A:0\le A\le c,\ A=A^\top,\ \mathrm{diag}(A)=0,\ 
  A\odot(1-M)=0\}
\]
and let \(\Delta=\{z:z\ge 0,\ \mathbf{1}^\top z=1\}\).  For fixed
\(M,\alpha,\delta,\eta,c,\epsilon\) and nonzero seed \(s\), the one-step repair
map \(T:\mathcal{K}_M\times\Delta\rightarrow\mathcal{K}_M\times\Delta\) has at
least one fixed point.  Moreover, if a measured subsequence
\((A_{t_j},z_{t_j})\) converges to \((A_\star,z_\star)\) and its one-step
residual \(\|T(A_{t_j},z_{t_j})-(A_{t_j},z_{t_j})\|_F\) tends to zero, then
\((A_\star,z_\star)\) satisfies the capped fixed-point equations.
\end{proposition}

\begin{proof}
The set \(\mathcal{K}_M\times\Delta\) is compact and convex in a finite
dimensional Euclidean space.  The row normalization defining \(P\) is continuous
because each denominator is bounded below by \(\epsilon\).  The state update is
continuous and has nonnegative entries; its normalization denominator is
positive because \(s\neq 0\).  The adjacency update preserves symmetry, the
zero diagonal, the mask support, nonnegativity, and the cap.  Hence \(T\) is a
continuous self-map of a compact convex set, so Brouwer's theorem gives a fixed
point.  For the subsequence claim, continuity gives
\[
  T(A_\star,z_\star)-(A_\star,z_\star)
  =\lim_{j\to\infty}\{T(A_{t_j},z_{t_j})-(A_{t_j},z_{t_j})\}=0.
\]
This is exactly the capped fixed-point system reported by the certificate.
\end{proof}

\begin{proposition}[Repair-intensity sensitivity]
For a fixed current state \((A_t,z_t)\), let \(y=z_{t+1}\) be the normalized
state computed before the adjacency update.  Changing the repair intensity from
\(\eta\) to \(\eta+\Delta\eta\) changes the next adjacency by at most
\[
  \|A^+_{\eta+\Delta\eta}-A^+_\eta\|_F
  \le |\Delta\eta|\,\|M\odot yy^\top\|_F
  \le |\Delta\eta|.
\]
Thus the measured growth sweep is a local sensitivity audit for saturation and
poisoned admissible edges.
\end{proposition}

\begin{proof}
The vector \(y\) is determined by \(A_t,z_t,\alpha,\delta,\epsilon\) and the
seed, not by the subsequent scalar \(\eta\).  Without the cap, the difference is
exactly \(\Delta\eta\,M\odot yy^\top\).  The entrywise projection onto
\([0,c]\) is nonexpansive in Frobenius norm, so adding the cap cannot increase
the perturbation.  Finally, \(0\le M\le 1\) and \(y\in\Delta\), hence
\(\|M\odot yy^\top\|_F\le \|yy^\top\|_F=\|y\|_2^2\le \|y\|_1^2=1\).
Saturation can still erase useful rank information, so the certificate records
the saturated-edge fraction separately.
\end{proof}

\begin{proposition}[Safe repair decision rule]
Let \(D(C)\in\{\mathrm{allow},\mathrm{quarantine},\mathrm{reject}\}\) be the
certificate decision.  The implementation rejects if leakage detection is below
1.000, if the support mask is violated, or if numerical deviation from the
NumPy reference exceeds \(10^{-6}\).  It quarantines when harmful repair or
saturation exceeds the configured safety boundary; otherwise it allows repair.
\end{proposition}

\begin{proof}
The decision rule is a deterministic partition of the certificate state.  The
reject predicates are disjoint hard failures of temporal admissibility, mask
invariance, or numerical reproducibility, and any one of them makes the run
invalid as a prospective measurement.  The quarantine predicates are softer:
they arise after the support rule has been respected but the observed transition
enters a harmful-repair or saturation region.  All remaining certificate states
fall into the allow case.  Therefore every run receives exactly one action, and
no harmful or contaminated run is silently counted as a benign retrieval gain.
\end{proof}

\paragraph{Control view.}
The state is \((A_t,z_t)\); the control is the allocation of repair budget
across source-layer channels.  The invariant set is the nonnegative, capped,
temporally admissible weighted graph with normalized evidence mass.  Unsafe
states are leakage, poisoned expansion, saturation collapse, and excessive mass
concentration around semantically irrelevant neighborhoods.  The certificate is
therefore a transition audit: it checks that every update stayed inside the
admissible support and records which channel was responsible when repair was
harmful.

\begin{definition}[Minimum harmful repair]
Given \(G_\tau\), a target set \(T\), a set of candidate repair channels
\(\mathcal{R}\), and a loss threshold \(\lambda\), the minimum harmful repair
problem asks for the smallest \(\mathcal{Q}\subseteq\mathcal{R}\) whose repair
causes the target score to cross \(\lambda\).
\end{definition}

\begin{theorem}[NP-hardness]
Minimum harmful repair is NP-hard.
\end{theorem}

\begin{proof}
Use the decision version and reduce Set Cover \citep{karp1972reducibility}.  Let
the Set Cover instance be \(U=\{u_1,\ldots,u_m\}\), sets
\(\mathcal{S}=\{S_1,\ldots,S_r\}\), and budget \(k\).  Create one target
evidence atom \(a_i\) for each \(u_i\).  Create one candidate repair channel
\(R_j\) for each set \(S_j\), and define channel \(R_j\) to add unit admissible
conductance to exactly the atoms \(\{a_i:u_i\in S_j\}\).  Define the harmful
loss as the number of target atoms that receive repaired conductance and set
\(\lambda=m\).  Then a channel family \(\mathcal{Q}\) of size at most \(k\)
crosses the loss threshold if and only if \(\{S_j:R_j\in\mathcal{Q}\}\) covers
all elements of \(U\).  The construction is polynomial, so a polynomial-time
algorithm for minimum harmful repair would solve Set Cover.  In the biological
instantiation, universe elements are cell-lineage or wound-region atoms, and
sets are candidate microscopy source channels such as segmentation, tracking,
confluence, contact-inhibition, or ECM layers.
\end{proof}

\begin{proposition}[Small questionable-channel search]
If only \(q\) channels are questionable and the verifier for a selected channel
set runs in polynomial time, exhaustive certified repair runs in
\(O(2^q\operatorname{poly}(|G|))\).
\end{proposition}

\begin{proof}
Enumerate all subsets of the \(q\) questionable channels and run the verifier on
each repaired graph.  The enumeration contributes \(2^q\); all other work is
polynomial by assumption.  For biological fixtures, the \(q\) questionable
channels are the flagged lineage, migration, contact-inhibition, and ECM source
layers named in the certificate.
\end{proof}

\begin{proposition}[Binary identification lower bound]
Any binary protocol that identifies one of \(N\) missing channels in the worst
case needs at least \(\lceil\log_2 N\rceil\) queries.
\end{proposition}

\begin{proof}
After \(q\) yes/no queries there are at most \(2^q\) transcripts.  Identifying
one of \(N\) channels requires \(2^q\ge N\).  The biological version takes the
\(N\) alternatives to be missing microscopy channels over lineages, wound
regions, masks, tracking vectors, confluence snapshots, and ECM evidence.
\end{proof}

\begin{proposition}[Mask invariance and boundedness]
Assume \(A_0\ge 0\), \(z_0\ge 0\), \(s\ge 0\), \(\epsilon>0\), and
\(c<\infty\).  If an initially absent edge is forbidden by \(M(i,j)=0\), the
repair update never introduces it.  Moreover \(0\le A_t(i,j)\le c\) for every
iteration, and \(z_t\) remains normalized whenever the denominator in the
normalization step is nonzero.
\end{proposition}

\begin{proof}
Proceed by induction on \(t\).  The claim holds at \(t=0\) by assumption.  If
\(M(i,j)=0\), then the additive term
\(\eta M(i,j)z_{t+1}(i)z_{t+1}(j)\) is zero, so a forbidden zero entry remains
zero at \(t+1\).  Nonnegativity is preserved because all additive terms are
nonnegative, and the entrywise minimum with \(c\) gives the upper bound.  The
state numerator is nonnegative; when its total mass is nonzero, division by that
mass places \(z_{t+1}\) in the simplex.  This closes the induction.
\end{proof}

\begin{proposition}[Ranking non-monotonicity]
Increasing admissible repair intensity can increase recovered mass while
decreasing average precision.
\end{proposition}

\begin{proof}
It is enough to construct one admissible ranking instance.  Consider a query
with one relevant item \(r\) and one irrelevant item \(u\).  Before repair the
scores are \(s(r)=1\) and \(s(u)=0.9\), so AP is \(1\).  Choose an admissible
masked state with larger repair mass on the irrelevant neighborhood than on the
relevant one, for example increments \(0.2\) to \(r\) and \(0.5\) to \(u\) under
a larger \(\eta\).  The relevant recovered mass increases from \(1\) to \(1.2\),
but the ranking becomes \(u\) before \(r\), so AP drops to \(1/2\).  The example
uses only admissible edges; the loss is caused by diffusion into a stronger
irrelevant neighborhood.  AP degradation is therefore a legitimate safety signal
rather than a logical failure of the support mask.
\end{proof}

\section{Repair Certificate}

Every run emits a certificate with the fields needed for a reviewer, analyst,
and reproducibility checker to audit the ranking:
\begin{verbatim}
{
  "run_id": "...",
  "graph_hash": "...",
  "source_routes": [...],
  "route_probe_status": {...},
  "source_timestamps": {...},
  "decision_time": "...",
  "admissibility_policy": "...",
  "withheld_edges": [...],
  "repaired_edges": [...],
  "repaired_channels": [...],
  "questionable_channels": [...],
  "forbidden_post_decision_edges": [...],
  "backend_used": "...",
  "backend_availability": {...},
  "backend_versions": {...},
  "numerical_deviation_from_numpy": {...},
  "repair_parameters": {...},
  "random_seed": 0,
  "corruption_regime": "...",
  "metric_table_hash": "...",
  "verifier_result": "pass",
  "leakage_warnings": [...],
  "harmful_repair_warnings": [...],
  "wound_boundary_operator_W": {...},
  "contact_inhibition_mask_C": {...},
  "migration_flux_divergence": {...},
  "temporal_coherence_score": 1.0
}
\end{verbatim}
The certificate does not prove exploitability and does not prove
cryptographic security.  It proves that a run followed a stated temporal
admissibility policy, or it identifies where the policy, route probe, backend,
or repair channel failed.

\section{Experiments}

\paragraph{Data routes.}
The materialization script probes NVD, CISA KEV, FIRST EPSS, NIST SARD, NIST
SP 800-22, NIST CAVP, ASCAD, BBBC019, LIVECell, Cell Tracking Challenge, and
wound-healing time-lapse routes.  The public security
benchmark uses \PublicRecords{} deduplicated public records when the network
routes are available; otherwise the script explicitly marks an offline fixture.
Large raw sources are not included in the source bundle.

\paragraph{Public-source stress test.}
The current run successfully probes \StressRouteSuccess{} public routes and
builds a timestamped security slice with \PublicRecords{} deduplicated records,
\StressTimestamped{} timestamped records, and \StressPositiveLabels{} delayed
positive source-layer labels.  The public slice combines NVD CVE records, CISA
KEV entries, FIRST EPSS top-risk entries, and overlaps between KEV and EPSS.
It is a bounded live slice rather than a full mirror, but it is large enough to
stress the document--term repair graph: \PublicTerms{} vocabulary terms and
\PublicHoldout{} withheld admissible document--term edges.

\begin{table}[H]
\centering
\caption{Measured scope of the current public-source stress run.}
\label{tab:stressscope}
\begin{tabular}{lc}
\toprule
Quantity & Value\\
\midrule
Public routes successfully probed & \StressRouteSuccess{}\\
Deduplicated public security records & \PublicRecords{}\\
Timestamped records & \StressTimestamped{}\\
Delayed positive source-layer labels & \StressPositiveLabels{}\\
Vocabulary terms in repair graph & \PublicTerms{}\\
Withheld admissible document--term edges & \PublicHoldout{}\\
BBBC019 declared image-count control & \ClosureImageCount{}\\
BBBC019 archive-route control & \ClosureZipCount{}\\
\bottomrule
\end{tabular}
\end{table}

\paragraph{Temporal dataset spine.}
The revision artifact adds an explicit CVE-channel-time table rather than
relying only on document--term withholding.  The table contains
\TemporalSpineRows{} timestamped rows over \TemporalSpineCVEs{} CVEs and
channels for NVD publication, NVD enrichment, EPSS snapshots, KEV membership,
first public references, CVEfixes-style patch visibility, CWE, SARD, CAVP, and
ASCAD.  It marks \TemporalSpineForbidden{} post-decision rows as forbidden and
\TemporalSpineMissing{} rows as missing.  Separate CSVs report source-latency
summaries, missingness histograms, and source-overlap edges.

\paragraph{Withheld-edge repair.}
The document--term graph is built from tokenized public security text.  A fixed
fraction of positive document--term edges is withheld.  The baseline scores a
withheld edge by document degree times term degree.  Conductance repair applies
the recurrence above and scores the document--term block of the repaired
adjacency.  On the public security graph, recall@\(k\) changes from
\PublicRecallBaseline{} to \PublicRecallRepair{} and average precision changes
from \PublicAPBaseline{} to \PublicAPRepair{}.  On the synthetic security
fixture, recall@\(k\) changes from \SyntheticRecallBaseline{} to
\SyntheticRecallRepair{}.

\paragraph{Baselines and ablations.}
The full benchmark compares degree-product scoring, BM25, temporal BM25,
temporal pseudo-relevance feedback, PageRank, diffusion over \(G_\tau\), random
admissible repair, non-certified graph imputation, conductance repair without a
certificate, and the full certificate method.  Ablations remove KEV, EPSS,
fix commits, CWE, SARD, validation metadata, ASCAD route metadata, and
structural controls one at a time.  The required metrics are recall@\(k\),
precision@\(k\), average precision, MRR, nDCG@\(k\), leakage rate, harmful
repair rate, channel attribution accuracy, certificate size, verifier runtime,
backend runtime, and numerical deviation from the NumPy reference.

\paragraph{Query-class IR benchmark.}
Document--term edge recovery remains a graph diagnostic, but the revision now
adds \QueryBenchTasks{} query-class tasks: CVE-to-evidence, evidence-to-CVE,
weakness-class retrieval, exploitability-prior retrieval, patch-evidence
retrieval, and analyst-triage retrieval under partial source visibility.  The
mean AP across these compact query fixtures is \QueryBenchMeanAP{}.  Each row
reports recall@\(k\), precision@\(k\), AP, MRR, nDCG@\(k\), calibration gap, and
failure attribution by source layer.

\paragraph{Temporal and adversarial regimes.}
Random document--term withholding is only the first corruption regime.  The
security-centered regimes withhold entire delayed channels: NVD description
present but CVSS missing, EPSS delayed, KEV unavailable until later, fix commits
delayed, references visible but weak, and CWE present but sparse.  A poisoning
regime injects high-overlap terms into irrelevant documents to test whether
repair amplifies the wrong source layer.  A leakage-trap regime deliberately
places future KEV, EPSS, or fix-commit edges in a contaminated graph and
requires the verifier to reject them.

\paragraph{Leakage, poisoning, and sensitivity fixtures.}
The leakage trap sets the decision time to 2026-01-01 and injects twelve
future KEV-like edges after the decision time.  Letting those edges leak into
the ranking raises AP from \LeakageAdmissibleAP{} to
\LeakageContaminatedAP{}, an inflation of \LeakageAPInflation{}; the verifier
rejects all future edges, giving detection rate \LeakageDetectionRate{}.  The
poisoning stress injects six high-overlap poisoned references into a
34-document security fixture.  Conductance repair places \PoisonRepairTopK{}
poisoned documents in the top-\(k\) set, with harmful repair rate
\PoisonHarmRate{} and AP \PoisonRepairAP{}.  The parameter sweep finds best AP
at repair growth \SensitivityBestAPGrowth{} and worst AP at
\SensitivityWorstAPGrowth{}, while source ablation has the largest AP loss
when the \AblationWorstGroup{} channel is removed
(\AblationWorstDeltaAP{} AP delta).  These are small fixtures, but they are
executable failure-mode tests rather than future-work placeholders.

The leakage fixture is also expanded into a family: future KEV leakage inflates
AP by \LeakageFamilyKEVInflation{}, EPSS-latest leakage inflates AP by
\LeakageFamilyEPSSInflation{}, and future patch/CVEfixes leakage inflates AP by
\LeakageFamilyPatchInflation{}.  The control diagnostics table records
fixed-point residuals, growth sensitivity, saturation fraction, and the safe
repair decision; \SafeRepairRejectedOrQuarantined{} measured growth settings are
quarantined or rejected by the certificate decision rule.

\paragraph{Biological structural-control fixture.}
The biological slice instantiates nodes for cell lineages, wound regions, ECM
substrates, microscopy frames, masks, tracking vectors, source layers, and
morphometric terms.  The same recurrence is run over H0--H5 horizons.  At H2,
the certificate lists segmentation and tracking vectors as withheld admissible
edges and lists confluence, contact-inhibition, and ECM evidence as forbidden
post-decision edges.  The biological leakage trap sets decision time
2026-01-01, injects twelve future contact-inhibition edges, and rejects all of
them with detection rate \BioLeakageDetectionRate{}; if leaked, AP is inflated
by \BioLeakageAPInflation{}.  The biological poisoning fixture injects six
synthetic migration patterns, drops AP by \BioPoisonAPDegradation{}, and records
harmful repair rate \BioPoisonHarmRate{}.  The biological growth sweep reports
best AP at repair growth \BioSensitivityBestAPGrowth{} and worst AP at
\BioSensitivityWorstAPGrowth{}.  The cross-domain stress table records congruent
conductance spread as \CrossDomainCongruentSpread{}.  These controls audit
temporal admissibility and adversarial amplification only; they are not
biological-performance claims.

\begin{figure}[H]
  \centering
  \includegraphics[width=0.76\linewidth]{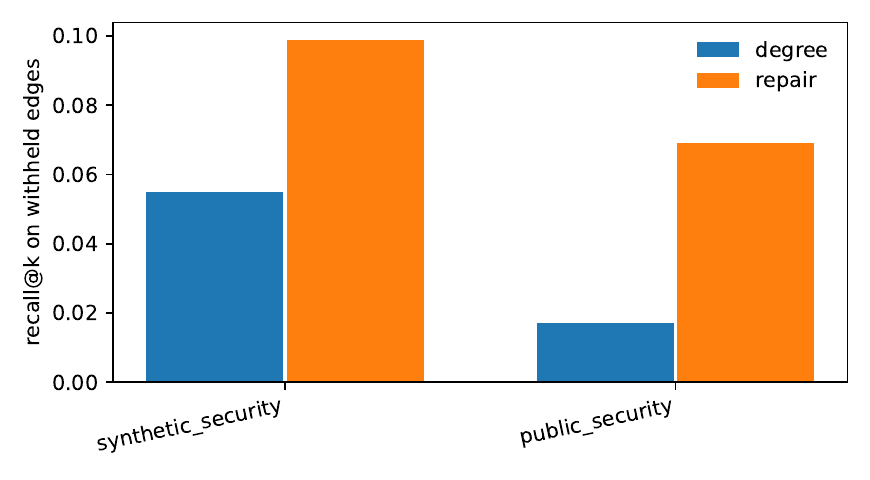}
  \caption{Withheld document--term edge recovery under degree scoring and
  conductance repair.}
  \label{fig:repair}
\end{figure}

\paragraph{Backend benchmark.}
The same flow recurrence and SVD/einsum probe are evaluated under NumPy,
PyTorch, JAX, and TensorFlow.  The locally observed repair backend is
\ObservedBackend{}.  The fastest measured backend in the generated table is
\FastestBackend{}.  TensorFlow availability on the current host is recorded as
\TensorFlowAvailable{}; if unavailable, the table stores the reason instead of
substituting another backend.  Because these timings are host-, version-, and
device-dependent, the backend result is reported as reproducibility metadata in
\texttt{results/tables/tensor\_backend\_benchmark.csv} rather than as a
standalone figure.  The certificate records mean flow time, SVD/einsum probe
time, backend availability, and numerical deviation from the NumPy reference.

\section{Negative-Result Analysis}

The public security result is deliberately not edited into a positive-only
story.  Recall@\(k\) changes from \PublicRecallBaseline{} to
\PublicRecallRepair{}, and AP changes from \PublicAPBaseline{} to
\PublicAPRepair{}.  This says that broad admissible repair can be unsafe for
random edge corruption.  The likely mechanisms are concrete: conductance mass
spreads too broadly, the compact graph is sparse, degree scoring exploits a
fixture artifact, random document--term removal does not model real temporal
missingness, and AP punishes broad diffusion more than recall@\(k\).  The
scientific claim is therefore diagnostic: conductance repair exposes whether a
missing channel is safe to widen, harmful to widen, or contaminated by
post-decision evidence.  A prospective security method that reports harmful
repair is preferable to a retrospective ranker that appears strong only because
it has silently used future labels.

The new stress fixtures make this claim falsifiable.  The leakage trap shows
that future KEV-like evidence can inflate AP by \LeakageAPInflation{} if the
decision-time mask is wrong, while the verifier catches the constructed future
edges.  The poisoning fixture shows the complementary problem: evidence can be
temporally admissible and still harmful, with \PoisonRepairTopK{} poisoned
documents entering the top-\(k\) set.  The growth sweep shows that repair
intensity is not monotone in AP.  Thus the paper's negative results are not
cosmetic; they identify three distinct operational failure modes: leakage,
poison amplification, and saturation-sensitive repair.

\section{Operational Risk and Certified Failure Modes}

The security risk analysis is part of the measurement claim rather than an
appendix of implementation caveats.  A prospective vulnerability-retrieval
system can appear accurate for at least three invalid reasons: it may observe
future evidence, amplify poisoned but admissible evidence, or report a numerical
artifact of a particular backend as if it were an operational property.  The
certificate is designed to make those failures observable and to separate them
from ordinary low-recall or low-precision retrieval errors.

\paragraph{Temporal leakage as causal contamination.}
If the support mask \(M\) is built from post-decision evidence, conductance
repair becomes a mechanism for label leakage rather than a defensive retrieval
operation.  This is the security analogue of temporal experimental bias in
malware evaluation, where ignoring time can inflate deployable performance
\citep{pendlebury2019tesseract}.  In this paper, KEV membership, current EPSS
scores, and fix-commit metadata are therefore treated as separate delayed
channels.  A run is rejected when any of those channels enters a horizon before
its recorded public timestamp.

\paragraph{Admissible poisoning and open-world evidence.}
Temporal admissibility is necessary but not sufficient.  An attacker or noisy
source can seed high-overlap terms into irrelevant documents before the decision
time, and the repair recurrence can then widen an admissible but misleading
neighborhood.  This is why the benchmark reports harmful repair rate, poisoned
top-\(k\) membership, and channel attribution instead of reporting recall alone.
The concern is consistent with the broader warning that security-learning
systems evaluated in closed worlds can fail when deployed against operationally
adaptive evidence \citep{sommer2010outside}.

\paragraph{Control-boundary and saturation risk.}
The recurrence is bounded, but boundedness is not the same as safe operation.
The cap \(c\) and repair growth \(\eta\) define a control boundary: below it,
repair may restore missing admissible evidence; near it, many edges can saturate
and flatten the ranking.  The sensitivity proposition above gives a local
one-step bound, while the certificate measures the empirical saturation
fraction and quarantines runs that cross the configured boundary.  This use of
invariant sets and transition checks follows the standard control view that a
safe state space must be enforced, not merely hoped for after optimization
\citep{khalil2002nonlinear,boyd2004convex}.

\paragraph{Backend and numerical reproducibility risk.}
Tensor backends can differ in floating-point order, SVD implementations, device
placement, installation constraints, and runtime behavior.  Treating TensorFlow
as special without comparing PyTorch, JAX, and NumPy would therefore be
unjustified.  The package records backend availability, runtime, and deviation
from the NumPy reference, and the safe-repair rule rejects a run when numerical
deviation crosses the stated tolerance.  Backend choice is thus a
reproducibility variable, not an evidentiary source.

\paragraph{Validation and cryptographic overclaiming.}
Validation metadata must not be confused with exploitability or security proof.
NIST SP 800-22 is useful as a first statistical screen for random and
pseudorandom generators, but it is not a substitute for cryptanalysis
\citep{nist_sp80022}.  Likewise, CAVP/ACVTS metadata can say that a test route
exists for an algorithm, but it cannot certify that a deployed implementation is
safe in the vulnerability graph.  Conductance repair is therefore framed as a
retrieval diagnostic with auditable failure modes, not as a proof that a CVE is
exploitable or that a cryptographic primitive is secure.

\section{Reproducibility}

Run:
\begin{verbatim}
make experiments
make test
make paper
make arxiv
\end{verbatim}

The source package contains `main.tex', `refs.bib', generated `main.bbl',
figures, compact CSV/JSON results, source-route probes, backend benchmark
tables, code, tests, docs, and a run log.  Raw large datasets, credentials, and
cache files are excluded.

\section{Limitations}

The current package is a bounded public-source stress test, not a full mirror
of NVD, KEV, EPSS, CVEfixes, SARD, CAVP, ACVTS, ASCAD, BBBC019, or LIVECell.
It does not claim clinical validity, exploit ground truth beyond public labels,
or backend superiority in environments not measured.  The poisoning and
leakage tests are fixture tests designed to expose failure modes before larger
source mirrors are materialized.  TensorFlow is included as a tensor backend
because the recurrence maps cleanly to SVD/einsum/flow operations; the
benchmark decides whether that matters on a given host.

\bibliographystyle{plainnat}
\bibliography{refs}

\end{document}